\documentclass[conference]{IEEEtran}
\usepackage{amssymb,amsmath,amsthm,mathrsfs}
\usepackage{graphics}

\usepackage[utf8]{inputenc}
\usepackage[english]{babel}
 \usepackage{hyperref}

\newtheorem{theorem}{Theorem}

\newtheorem{lemma}[theorem]{Lemma}
\newtheorem{prop}[theorem]{Proposition}
\theoremstyle{definition}\newtheorem{definition}[theorem]{Definition}
\theoremstyle{remark}\newtheorem{remark}[theorem]{Remark}

\ifCLASSINFOpdf
\else
\fi
\begin{document}
%
\title{On The Construction of Capacity-Achieving Lattice Gaussian Codes}

\author{\IEEEauthorblockN{Wael Alghamdi, Walid Abediseid, and Mohamed-Slim Alouini}
\IEEEauthorblockA{Computer, Electrical and Mathematical Sciences and Engineering (CEMSE) Division,\\ King Abdullah University of Science and Technology (KAUST),Thuwal, Makkah Province, Saudi Arabia,\\E-mail: wael.alghamdi, walid.abediseid, slim.alouini@kaust.edu.sa}
}


%


\maketitle

\begin{abstract}
In this paper, we propose a new approach to proving results regarding channel coding schemes based on construction$-$A lattices for the Additive White Gaussian Noise (AWGN) channel that yields new characterizations of the code construction parameters, i.e., the primes and dimensions of the codes, as functions of the block-length. The approach we take introduces an averaging argument that explicitly involves the considered parameters. This averaging argument is applied to a generalized Loeliger ensemble \cite{Loeliger} to provide a more practical proof of the existence of AWGN-good lattices, and to characterize suitable parameters for the lattice Gaussian coding scheme proposed by Ling and Belfiore \cite{Ling_Belf}.
\end{abstract}


%
\IEEEpeerreviewmaketitle

\section{Introduction}
An explicit construction of a structured coding scheme that achieves the capacity of the additive white Gaussian noise (AWGN) channel has been a major problem in coding theory lately. Shannon first proved, via averaging over all possible codebooks of a certain blocklength, that there are coding schemes that achieve the capacity of the AWGN channel \cite{Shannon}. In \cite{Poltyerv}, Poltyerv showed, via an averaging argument, that linear codes can achieve the capacity of the unconstrained AWGN channel. A new line of study was initiated in \cite{Loeliger} when Loeliger showed that construction$-$A lattices can be made to behave like a Minkowski-Hlawka-Siegel (MHS) ensemble. This was used by Erez and Zamir in \cite{Erez_Zamir} to show that nested construction$-$A lattices with dithering can achieve the capacity of the  AWGN channel, and by Ling and Belfiore in \cite{Ling_Belf} to show that a lattice Gaussian coding scheme based on construction$-$A lattices can achieve the capacity of the AWGN channel without the need of dithering. However, a practical piece of the puzzle remains missing, which is the treatment of the parameters defining the construction$-$A lattices used in the coding schemes. The work by Loeliger \cite{Loeliger} uses the property that construction$-$A lattices can be made to behave like an MHS ensemble as an input to the MHS theorem, which destroys the explicitness of the parameters involved. 

In this paper, we resolve this issue by refraining from using the MHS theorem. We show that asymptotic results regarding a Riemann theta function and a Pochhammer symbol suffice to get stronger versions of the previously known results and new characterizations of the primes and dimensions of the construction$-$A lattices that are used to build capacity-achieving codes.

We use the following notations. The symbol $\log$ always refers to the natural logarithm, and information is measured in nats. For any set $S,$ $|S|$  denotes the number of elements in $S,$ $1_S$ the indicator function of $S$ and $\mathcal{P}(S)$  the power set of $S.$ The notation $\|\cdot\|$ will always refer to the $2-$norm. The symbol $0$ will refer to either a scalar (in $\mathbb{R}$ or $\mathbb{F}_p$), a vector (in $\mathbb{R}^n$ or $\mathbb{F}_p^n$) or a matrix (over $\mathbb{R}$ or $\mathbb{F}_p$), but it will be clear from context which is the meaning referred to. We will use $\mu_L$ to refer to the Lebesgue measure over $\mathbb{R}^n$ for any fixed $n,$ which will be clear from the context. Also, for any natural $n,$ point $q\in \mathbb{R}^n$ and $r>0,$ we will denote by $\mathcal{B}_n(q,r)$ the open ball in $\mathbb{R}^n$ of radius $r$ around $q.$

\section{Preliminaries}

We develop in this section the mathematical tools we  need.

\subsection{Lattices and Lattice Ensembles}

A lattice in $\mathbb{R}^n$ is a set $\Lambda =\lbrace B x \; ; \; x \in \mathbb{Z}^n \rbrace,$ where $B \in \mathbb{R}^{n\times n}$ is full-rank. To any lattice $\Lambda$ in $\mathbb{R}^n,$ one may associate the (uniformly convergent over every $[\delta,\infty)\subset (0,\infty)$) theta series $\Theta_{\Lambda}\left(\tau\right) := \sum_{\lambda \in \Lambda} e^{-\pi \tau\|\lambda\|^2}.$ We denote a fundamental Voronoi region of $\Lambda$ by $\mathcal{V}(\Lambda)\subset \mathbb{R}^n$ (which differs from the set $\lbrace y\in \mathbb{R}^n \; ; \; \min_{\lambda \in \Lambda} \|y-\lambda\|=\|y\|\rbrace$ by a set of measure $0$), and the dual lattice by $\Lambda^*.$ The following is a classical result.

\begin{theorem}[Theta Series Functional Equation] \label{LT} For any lattice $\Lambda\subset \mathbb{R}^n$ and any $t>0,$
$$
\Theta_\Lambda(t) = t^{-n/2}\mu_L(\mathcal{V}(\Lambda))^{-1}\Theta_{\Lambda^*}(t^{-1}).
$$
\end{theorem}
Since an integer lattice $\mathbb{Z}^n$ is self-dual and satisfies $\mu_L(\mathcal{V}(\mathbb{Z}^n))=\mu_L([-1/2,1/2]^n)=1,$ theorem \ref{LT} yields that, for any positive integer $n$ and positive real $t,$
\begin{equation}\label{TS}
\Theta_{\mathbb{Z}^n}(t)=t^{-n/2}\Theta_{\mathbb{Z}^n}(t^{-1})
\end{equation}

A linear code is a set $C(M):=\lbrace Mx \; ; \; x\in \mathbb{F}_q^{k} \rbrace$ where $q$ is a prime power and $ M\in \mathbb{F}_q^{n\times k}.$ If $p$ is prime and $C(M) \subset \mathbb{F}_p^n$ is a linear code, one may show that the Minkowski sum $\Lambda(M):=C(M)+p\mathbb{Z}^n$ is a lattice. Such a lattice is called a construction$-$A, or mod$-p$ lattice. 

Let $\mathscr{P}$ denote the set of prime numbers. For any integer $n\geq 2$ and $(k,p,a)\in \lbrace 1,\cdots,n-1 \rbrace\times \mathscr{P}\times \mathbb{R}_{>0},$ we call $(n,k,p,a)$ a \textit{quadruple of parameters}, and we denote it usually by $\mathfrak{p}.$ For any quadruple of parameters $\mathfrak{p}=(n,k,p,a),$ denote $V_{\mathfrak{p}}=a^np^{n-k}.$ Note that, if $1\leq k \leq n$ and $M\in \mathbb{F}_p^{n\times k}$ is full-rank, then $\mu_L(\mathcal{V}(a\Lambda(M))) = V_{(n,k,p,a)}.$

For any quadruple of parameters $\mathfrak{p}=(n,k,p,a)$  and random variable $G$ over $\mathbb{F}_p^{n\times k},$ we use the following notation. Let  $M_\mathfrak{p}\subset \mathbb{F}_p^{n\times k}$ denote the subset of all full-rank matrices. Define $U_\mathfrak{p}'$ and $U_\mathfrak{p}$ to be random matrices uniformly distributed over $\mathbb{F}_p^{n\times k}$ and $M_\mathfrak{p},$ respectively, and $u_\mathfrak{p}$ to be a random vector uniformly distributed over $\mathbb{F}_p^k.$ One may also consider the random lattice $\Lambda(G)$ (see part \hyperref[AC]{1} of Appendix A). We denote $\Lambda_\mathfrak{p}'=a\Lambda(U_\mathfrak{p}')$ and $\Lambda_\mathfrak{p}=a\Lambda(U_\mathfrak{p})$ for short (note that $\Lambda_\mathfrak{p}'$ is a Loeliger ensemble). We set $\xi^{\max}(G)=\max_{y\in \mathbb{F}_p^{n}\setminus \lbrace 0 \rbrace}\mathrm{Pr}(Gu_\mathfrak{p}=y)$ and $ \xi^{(0)}(G)=\mathrm{Pr}(Gu_\mathfrak{p}=0).$ Denote $\xi_\mathfrak{p}=\xi^{(0)}(U_\mathfrak{p}')$ for short, and note that $\xi^{\max}(U_\mathfrak{p}')=\frac{1-\xi_\mathfrak{p}}{p^n-1}.$ Also, for any $M\in \mathbb{F}_p^{n\times k},$ we have that $M0=0,$ so $1/p^k\leq \xi_\mathfrak{p}.$

\subsection{An Averaging Argument}

The following inequality is used to derive an averaging argument for lattice sums in proposition \ref{av}.

\begin{lemma} \label{ineq} For any quadruple of parameters $\mathfrak{p}=(n,k,p,a),$ $M\in \mathbb{F}_p^{n\times k}$ and $s:\mathbb{R}^n\longrightarrow [0,\infty],$ we have that
$$
\sum_{\lambda \in \Lambda(M)}  s(\lambda) \leq p^{k} \sum_{y\in \mathbb{F}_{p}^{n}}\sum_{z\in \mathbb{Z}^n}  \mathrm{Pr}(Mu_{\mathfrak{p}}=y) \cdot s(y+pz).
$$
\end{lemma}
\begin{proof}
This follows from
\begin{align*}
\sum_{\lambda \in \Lambda(M)}  s(\lambda) =\sum_{z\in \mathbb{Z}^n} \sum_{y\in \mathbb{F}_{p}^{n}}  s(y+pz) \cdot 1_{C(M)}(y)
\end{align*}
and $1_{C(M)}(y)\leq |\lbrace x\in \mathbb{F}_{p}^{k}  ; Mx=y \rbrace| = p^{k}  \mathrm{Pr}(Mu_{\mathfrak{p}}=y).$
\end{proof}

\begin{prop} \label{av} For any quadruple of parameters $\mathfrak{p}=(n,k,p,a),$ random variable $G$ over $\mathbb{F}_p^{n\times k}$ and $g:\mathbb{R}^n\longrightarrow [0,\infty],$ we have that 
$$
\mathbb{E}_{G}\left[ \sum_{\lambda \in\Lambda(G)} g(\lambda) \right] \leq p^k \mathbb{E}_{Gu_\mathfrak{p}}\left[ \sum_{z\in \mathbb{Z}^n} g(Gu_\mathfrak{p}+pz) \right].
$$
\end{prop}

\begin{proof} 
See Appendix \hyperref[B]{B}.
\end{proof}
 
The following proposition applies the averaging argument in proposition \ref{av} on a counting function that we define now. For any $n\in \mathbb{Z}_{>0}$ and  $S\subset\mathbb{R}^n,$ define $N_S: \mathcal{P}(\mathbb{R}^n) \rightarrow \mathbb{Z}_{\geq 0} \cup \lbrace \infty \rbrace$ by $N_S(\Lambda)=|\Lambda \cap (S\setminus \lbrace 0 \rbrace)|.$ 

\begin{prop} \label{Av1}
For any quadruple of parameters $\mathfrak{p}=(n,k,p,a),$ random variable $G$ over $\mathbb{F}_p^{n\times k}$  and $S\subset \mathbb{R}^n,$ we have that 
\begin{equation}
 \mathbb{E}_{G}\left[ N_{S}(a\Lambda(G)\setminus ap\mathbb{Z}^n) \right] \leq p^k \cdot \xi^{\max}(G) \cdot  N_{S}(a\mathbb{Z}^n ) \label{EZNS}
\end{equation}
and
\begin{equation}
 \mathbb{E}_{U_\mathfrak{p}}\left[ N_{S}(\Lambda_\mathfrak{p}\setminus ap\mathbb{Z}^n) \right] \leq \frac{p^k(1-\xi_\mathfrak{p})}{(1-p^{k-n})(p^n-1)} \cdot  N_{S}(a\mathbb{Z}^n ) \label{EUNS}
\end{equation}
\end{prop}

\begin{proof}
See Appendix \hyperref[B]{B}.
\end{proof}

Inequality \ref{EUNS} will be used to prove lemma \ref{IP}, thereby giving an upper bound on the probability of error for lattice decoding. Another application of the averaging argument is deriving an upper bound on the flatness factor in proposition \ref{FF1}.

Before turning to the probability of error of lattice decoding, we mention a few properties of the counting function $N_S.$ 

First, for any $S,\Lambda \subset \mathbb{R}^n$ and $a>0,$ it is clear that $N_S(a\Lambda)=N_{\frac{1}{a}S}(\Lambda).$ Moreover, the following two lemmas are useful.

\begin{lemma}[\cite{Or_Erez}] \label{OR}
 For any $S,\Lambda\subset \mathbb{R}^n,$ $q\in \mathbb{R}^n$ and $r>0,$ we have that $N_{\mathcal{B}_n(q,r)}\left( \mathbb{Z}^n\right)\leq \mu_L(\mathcal{B}_n(0,1))(r+\sqrt{n}/2)^n.$
\end{lemma}

\begin{lemma} \label{NS}
    For any quadruple of parameters $\mathfrak{p}=(n,k,p,a)$ and  $r>0,$ we have that 
    $$
    \left\lbrace q\in \mathbb{R}^n \; ; \; N_{\mathcal{B}(q,r)}(ap\mathbb{Z}^n)\geq 1\right\rbrace \subset \left\lbrace q\in \mathbb{R}^n \; ; \; \|q\|\geq ap-r \right\rbrace.
 $$
 Also, with $D=\lbrace 0 \rbrace \cup [1,\infty],$ any $D-$valued random variable $L$ satisfies $\mathrm{Pr}( L\geq 1)\leq \mathbb{E}[L].$ 
   \end{lemma}
   \begin{proof}
   See Appendix \hyperref[B]{B}.
   \end{proof}

\subsection{The Probability of Error}

The tools we develop in this section will be used in theorem \ref{PE} to get an upper bound on the probability of error of lattice decoding involving the Poltyrev exponent.   

Throughout the paper, we fix a sequence $\lbrace \sigma_{w,n} \rbrace \subset \mathbb{R}_{>0},$ and for each $n,$ we let $W^{(n)}$ denote a random vector whose components are i.i.d. zero-mean Gaussian random variables of variance $\sigma_{w,n}^2.$  We also denote the probability density function of a random variable $Z$ by $f_Z.$

A useful result used to derive error exponents is the following version of the Chernoff bound.

\begin{lemma}[Chernoff Bound, proposition 13.1.3 in \cite{Zamir}] \label{CB}
 For any $r>0,$ if  $E_{\mathrm{sp}}(x) := 1_{[1,\infty)}(x)\cdot (x-1-\ln x)/2,$ then
$$
\mathrm{Pr}(\|W^{(n)}\|>r ) \leq \exp\left(-nE_{\mathrm{sp}}\left( \frac{r^2}{n\sigma_{w,n}^2} \right)\right).
$$
\end{lemma}

Recall that, for any lattice $\Lambda$ in $\mathbb{R}^n,$ the probability of error for lattice decoding in the presence of noise $W^{(n)}$ is given by $\mathrm{Pr}(W^{(n)}\not\in \mathcal{V}(\Lambda))=\mathrm{Pr}(N_{\mathcal{B}_n(W^{(n)},\|W^{(n)}\|)}(\Lambda)\geq 1).$ 

Define, for any quadruple of parameters $\mathfrak{p}=(n,k,p,a)$ and any random variable $G$ over $\mathbb{F}_p^{n\times k},$   $h(\mathfrak{p},G,\rho):= \mathbb{E}_{W^{(n)}}\left[ \mathbb{E}_{G}\left[ N_{\mathcal{B}_n(W^{(n)},\rho)}(a\Lambda(G)\setminus ap\mathbb{Z}^n)\right] \; \left| \; \|W^{(n)}\|=\rho\right]\right.$ (see part \hyperref[AI]{3} of Appendix A),
\begin{align*}
 I_\mathfrak{p}\left(G,W^{(n)} \right) := \int_0^\infty f_{\|W^{(n)}\|}(\rho) \cdot \min\left( h(\mathfrak{p},G,\rho),1 \right) \, d\rho,
 \end{align*}
and
\begin{equation}\label{ANN}
A_{\mathfrak{p}}^{\mathrm{NN}}(G,W^{(n)}) := \mathrm{Pr}( \|W^{(n)}\|>ap/2 ) + I_\mathfrak{p}\left(G,W^{(n)}\right)
\end{equation}
Denote $B_\mathfrak{p}^{\mathrm{NN}}(W^{(n)})=A_{\mathfrak{p}}^{\mathrm{NN}}(U_{\mathfrak{p}},W^{(n)}).$ Since $\mathbb{E}_G$ is a finite linear combination, and since the counting function $N_S$ is always nonnegative, one may exchange the order of expectations in the definition of $h(\mathfrak{p},G,\rho),$ i.e., we may rewrite $h(\mathfrak{p},G,\rho)=\mathbb{E}_{G}\left[ \mathbb{E}_{W^{(n)}}\left[ N_{\mathcal{B}(W^{(n)},\rho)}(a\Lambda(G)\setminus ap\mathbb{Z}^n) \; \left| \; \|W^{(n)}\|=\rho\right]\right.\right].$

\begin{prop} \label{PE1} For any quadruple of parameters $\mathfrak{p}=(n,k,p,a)$ and any random variable $G$ over $\mathbb{F}_p^{n\times k},$ we have that
     $$
      \mathbb{E}_{Z}\left[ \mathrm{Pr}\lbrace W^{(n)}\not\in \mathcal{V}(a\Lambda(G)) \rbrace \right] \leq A_{\mathfrak{p}}^{\mathrm{NN}}(G,W^{(n)}).
     $$
\end{prop}
     
\begin{proof}  For any subset $S\subset \mathbb{R}^n$ and any $r>0,$ denote $g(S,r)=\mathrm{Pr}(N_{\mathcal{B}(W^{(n)},r)}(S)\geq 1 \; | \; \|W^{(n)}\|=r ).$ Then, 
\begin{equation*} 
\mathrm{Pr}\lbrace W^{(n)}\not\in \mathcal{V}(a\Lambda(G)) \rbrace = \int_0^\infty f_{\|W^{(n)}\|}(r) \cdot g(a\Lambda(G),r)\, dr
\end{equation*}
Since $g(S\cup T,r)=g(S,r)+g(T,r)$ whenever $S$ and $T$ are disjoint sets, we see that
     \begin{align}
     \mathrm{Pr}\lbrace W^{(n)}\not\in &\mathcal{V}(a\Lambda(G)) \rbrace = \int_0^\infty f_{\|W^{(n)}\|}(r) \cdot g(ap\mathbb{Z}^n,r)\, dr \nonumber \\
     &+ \int_0^\infty f_{\|W^{(n)}\|}(r) \cdot g(a\Lambda(G)\setminus ap\mathbb{Z}^n,r) \, dr \label{ER}
     \end{align}
     We will upper bound  the first integral in equation \ref{ER} by $\mathrm{Pr}\lbrace \|W^{(n)}\|>ap/2 \rbrace,$ and the expectation, with respect to $G,$ of the second integral in \ref{ER} by $I_\mathfrak{p}(G,W^{(n)}).$ 
     
     The first part of lemma \ref{NS} yields the estimate $g(ap\mathbb{Z}^n,r) \leq \mathrm{Pr}(\|W^{(n)}\|\geq ap-r \; | \; \|W^{(n)}\|=r),$ so $g(ap\mathbb{Z}^n,r)\leq 1_{[ap/2,\infty)}(r).$ Hence,
     \begin{align*}
     \int_0^\infty f_{\|W^{(n)}\|}(r)   \cdot g(ap\mathbb{Z}^n,r) \, dr \leq \mathrm{Pr}\left(\|W^{(n)}\|\geq ap/2 \right) 
     \end{align*}
    The second part of lemma \ref{NS} yields $g(a\Lambda(G)\setminus ap\mathbb{Z}^n,r)\leq \mathbb{E}_{W^{(n)}}\left[ N_{\mathcal{B}(W^{(n)},r)}(a\Lambda(G)\setminus ap\mathbb{Z}^n) \; \left| \; \|W^{(n)}\|=r\right]\right.,$ so $\mathbb{E}_G\left[ g(a\Lambda(G)\setminus ap\mathbb{Z}^n,r) \right] \leq h(\mathfrak{p},G,r).$ Then, 
     \begin{align*}
    &\mathbb{E}_{G}\left[\int_0^\infty f_{\|W^{(n)}\|}(r) \cdot g(a\Lambda(G)\setminus ap\mathbb{Z}^n,r) \, dr\right] \nonumber \\
    &= \int_0^\infty f_{\|W^{(n)}\|}(r) \cdot  \mathbb{E}_G\left[ g(a\Lambda(G)\setminus ap\mathbb{Z}^n,r)\right] \, dr \nonumber \\ 
     &\leq \int_0^\infty f_{\|W^{(n)}\|}(r) \cdot \min\left( h(\mathfrak{p},G,r),1 \right) \, dr   = I_\mathfrak{p}(G,W^{(n)}),
     \end{align*}
     as desired.
   \end{proof}  

Recall that the unexpurgated Poltyrev exponent is given by
   $$
   E_P^{\mathrm{un}}(b) = \left\lbrace \begin{array}{ll}
   E_{\mathrm{sp}}(b) &, \text{ if } 1\leq  b <2 \\
   \frac{1}{2}\log \frac{eb}{4} &, \text{ if } 2\leq b.
   \end{array} \right.,
   $$
   and the Volume-to-Noise Ratio (VNR) of a lattice $\Lambda \subset \mathbb{R}^n$ is defined by $\gamma_{\Lambda}(\sigma)= \mu_L(\mathcal{V}(\Lambda))^{2/n}/\sigma^2.$ 
   
By bounding $\mathrm{Pr}( \|W^{(n)}\|>ap/2 )$ using the Chernoff bound, and $I_\mathfrak{p}\left(U_{\mathfrak{p}},W^{(n)}\right)$ as in the following lemma, one might be able to make $B_{\mathfrak{p}}^{\mathrm{NN}}(W^{(n)})$ vanish as $\exp\left(-nE_P^{\mathrm{un}}\left( \frac{\gamma_{\Lambda_\mathfrak{p}}(\sigma_{w,n})}{2\pi e } \right)\right).$ Theorems \ref{PE} and \ref{CM} discuss this. 

    \begin{lemma} \label{IP}
   Let $\mathfrak{p}=(n,k,p,a)$ be a quadruple of parameters such that $\varepsilon:= V_{\mathfrak{p}}^{2/n}/(2\pi e \sigma_{w,n}^2)-1>0.$ Then,
  \begin{align*}
    &I_{\mathfrak{p}}\left(W^{(n)},\Lambda_{\mathfrak{p}}\right)\leq \left(\frac{\sqrt{\pi e/2}}{p^{1-k/n}} \right)^n \\
    &+n\left( \left( \frac{1}{\sqrt{1+\varepsilon}}+ \frac{\sqrt{\pi e/2}}{p^{1-k/n}} \right)^{n} +e^{-n \cdot \inf_{u\in C}v_{\mathfrak{p}}^{\mathrm{NN}}(u,\varepsilon)}\right),
    \end{align*}
   where 
   \begin{align*}
  v_{\mathfrak{p}}^{\mathrm{NN}}(u,\varepsilon) &:=   E\left( \frac{u^2}{n\sigma_{w,n}^2} \right) \\
  &\;\; - \frac{n-1}{n}\log \left( \frac{u}{\sqrt{n\sigma_{w,n}^2(1+\varepsilon)}} + \frac{\sqrt{\pi e/2}}{p^{1-k/n}} \right)
 \end{align*}
 and $C:=\left[ \sqrt{n\sigma_{w,n}^2}, \sqrt{n\sigma_{w,n}^2(1+\varepsilon)}\right].$ 
   \end{lemma}
   \begin{proof}
   See Appendix \hyperref[C]{C}.
   \end{proof}

    A relation between $ \inf_{u}v_{\mathfrak{p}}^{\mathrm{NN}}(u,\varepsilon)$ and Poltyrev's unexpurgated error exponent is given in the following lemma.

    \begin{lemma} \label{VNN}
   Let $\lbrace \mathfrak{p}_n=(n,k_n,p_n,a_n) \rbrace_{n\in \mathbb{Z}_{>1}}$ be a sequence of quadruples of parameters, and $b>0.$
   Assume that $\lim_{n\rightarrow \infty} p_n^{1-k_n/n}=\infty.$ Then, as $n\longrightarrow \infty,$  $\inf_u v_{\mathfrak{p}_n}^{\mathrm{NN}}(u,b_n)=E_{P}^{\mathrm{un}}(1+b)+o(1).$
    \end{lemma}
    \begin{proof}
   See Appendix \hyperref[C]{C}.
    \end{proof}
    
    \subsection{The Flatness Factor}
    
Ling and Belfiore define the flatness factor $\epsilon_{\Lambda}:\mathbb{R}_{>0}\longrightarrow [0,\infty)$ of a lattice $\Lambda \subset \mathbb{R}^n$ in \cite{Ling_Belf} and derive the expression
\begin{equation} \label{FD}
\epsilon_{\Lambda}(\sigma)=  \frac{\mu_L(\mathcal{V}(\Lambda))}{(2\pi \sigma^2)^{n/2}} \Theta_{\Lambda}\left( \frac{1}{2\pi \sigma^2} \right) - 1
\end{equation}
It is desirable, for lattice Gaussian coding, to have this flatness factor be small. Proposition \ref{FF1} will give an estimate on the average size of the flatness factor, and theorem \ref{FF} will give conditions under which this estimate is small.

   First, let us define the estimates that will be used in proposition \ref{FF1} . For any $\tau>0,$ quadruple of parameters $\mathfrak{p}=(n,k,p,a)$ and random variable $G$ over $\mathbb{F}_p^{n \times k},$ define
\begin{align} 
A^{\mathrm{Fl}}_\mathfrak{p}(G,&W^{(n)},\tau) = p^k \xi^{\max}(G) \cdot \Theta_{\mathbb{Z}^n}(a^2\tau)  \nonumber \\
&+ p^k(\xi^{(0)}(G)-\xi^{\max}(G))\cdot \Theta_{\mathbb{Z}^n}(a^2p^2\tau)  \label{AFL}
\end{align}
and $B^{\mathrm{Fl}}_\mathfrak{p}(W^{(n)},\tau) := V_{\mathfrak{p}} \tau^{n/2} A_{\mathfrak{p}}^{\mathrm{FL}}(U_{\mathfrak{p}},W^{(n)},\tau)/(1-p^{k-n}).$

By equation \ref{TS}, we may rewrite 
\begin{align}
A^{\mathrm{Fl}}_\mathfrak{p}&(G,W^{(n)},\tau) =V_{\mathfrak{p}}^{-1}p^n  \tau^{-n/2}\xi^{\max}(G) \cdot \Theta_{\mathbb{Z}^n}\left( \frac{1}{a^2 \tau} \right) \nonumber \\
&+ p^k(\xi^{(0)}(G)-\xi^{\max}(G))\cdot \Theta_{\mathbb{Z}^n}(a^2p^2\tau) \label{AFL1}
\end{align}

It is useful to recall the following lemma. 
  \begin{lemma} \label{U}
  For any quadruple of parameters $\mathfrak{p}=(n,k,p,a)$ and $f:\mathbb{F}_p^{n\times k} \longrightarrow [0,\infty],$ we have that
  $$
  \mathbb{E}_{U_\mathfrak{p}'}\left[ f(U_\mathfrak{p}')\right] \geq  \mathbb{E}_{U_\mathfrak{p}}\left[ f(U_\mathfrak{p})\right](1-p^{k-n}).
  $$
  \end{lemma}
  \begin{proof}
   See Appendix \hyperref[E]{E}.
  \end{proof}
  
  The second application of the averaging argument is via considering a Gaussian function.

 \begin{prop} \label{FF1} For any quadruple of parameters $\mathfrak{p}=(n,k,p,a),$ random variable  $G$ over $\mathbb{F}_p^{n\times k},$ and $\tau>0,$ we have that
 \begin{equation}
  \mathbb{E}\left[ \Theta_{a\Lambda(G)}(\tau)\right] \leq  A^{\mathrm{Fl}}_\mathfrak{p}(G,W^{(n)},\tau) \label{FF11}
  \end{equation}
and
 \begin{equation}
  \mathbb{E}\left[ \epsilon_{\Lambda_\mathfrak{p}}\left(\frac{1}{\sqrt{2\pi \tau}} \right)\right] \leq  B^{\mathrm{Fl}}_\mathfrak{p}(W^{(n)},\tau) \label{FF12}
 \end{equation}
 \end{prop}
 \begin{proof}
 Define $g: \mathbb{R}^n \longrightarrow [0,\infty]$  by $g(\lambda)=e^{-\pi \tau \|a\lambda\|^2}.$ Then, $\sum_{\lambda \in \Lambda(G)} g(\lambda)=\Theta_{a\Lambda(G)}(\tau).$ Hence, 
\begin{align*}
&\mathbb{E}_G\left[ \Theta_{a\Lambda(G)}(\tau) \right] \leq p^k \mathbb{E}_{Gu_\mathfrak{p}}\left[ \sum_{v\in \mathbb{Z}^n} e^{-\pi \tau \|a(Gu_\mathfrak{p}+pv)\|^2} \right] \\
&\leq  {p}^{k} \left( \xi^{\max}(G)\sum_{y \in \mathbb{F}_{p}^n} \sum_{v\in \mathbb{Z}^n}e^{-a^2\pi \tau \|y+pv\|^2} \right. \\
&\;\; + \left. \left( \xi^{(0)}(G)-\xi^{\max}(G) \right)\sum_{t\in \mathbb{Z}^n}e^{-a^2\pi \tau \|pt\|^2} \right)  \\
&= p^k \xi^{\max}(G)  \Theta_{\mathbb{Z}^n}(a^2\tau) + p^k(\xi^{(0)}(G)-\xi^{\max}(G)) \Theta_{\mathbb{Z}^n}(a^2p^2\tau)
\end{align*}
which is just $A^{\mathrm{Fl}}_\mathfrak{p}(G,W^{(n)},\tau)$ by equation \ref{AFL}. Finally, using equation \ref{AFL1} instead, substituting $G=U_\mathfrak{p}',$ and combining equation \ref{FD} and lemma \ref{U}, one gets inequality \ref{FF12}.
 \end{proof}
 
  The following two lemmas give  asymptotic formulas that will be helpful in theorem $\ref{FF}.$
  
  \begin{lemma} \label{PH}
Let  $\lbrace \mathfrak{p}_n=(n,k_n,p_n,a_n) \rbrace_{n\in \mathbb{Z}_{>1}}$ be a sequence of quadruples of parameters. If $\lim_{n\rightarrow \infty} p_n^{n-k_n}=\infty,$ then $\lim_{n\rightarrow \infty} \xi_{\mathfrak{p}_n}p_n^{k_n} = 1.$ 
\end{lemma}
\begin{proof}
   See Appendix \hyperref[D]{D}.
\end{proof}

\begin{lemma}\label{RT}
 For any sequence $\lbrace c_n\rbrace_{n\in \mathbb{N}}\subset \mathbb{R}_{>0},$ we have that
$\lim_{n\rightarrow \infty} \Theta_{\mathbb{Z}^n}(c_n) =1$ if and only if $n=o(e^{\pi c_n})$ as $n\longrightarrow \infty.$ 
\end{lemma}
\begin{proof}
   See Appendix \hyperref[F]{F}.
\end{proof}

  \section{Construction Parameters}
  
  In this section, a wide range of quadruples of parameters are shown to yield reliable and capacity achieving coding.
  
\subsection{Nearest-Neighbor Decoding} 
     
     The following theorem shows that primes of size at least comparable to the square root of the block length make lattice decoding reliable.
     
 \begin{theorem} \label{PE} Let $\lbrace\mathfrak{p}_n = (n,k_n,p_n,a_n)\rbrace_{n\in\mathbb{Z}_{>1}}$ be a sequence of quadruples of parameters, and $\delta>1.$ If we have that, for each $n,$ $\varepsilon:= V_{\mathfrak{p}_n}^{2/n}/(2\pi e \sigma_{w,n}^2)-1>0$ is constant and 
 $$
 p_n>\left( \frac{2\delta n}{\pi e (1+\varepsilon)} \right)^{n/(2k_n)},
 $$
  and if $\lim_{n\rightarrow \infty} p_n^{1-k_n/n} = \infty,$ then, as $n\longrightarrow \infty,$
 \begin{equation} 
 B_{\mathfrak{p}_n}^{\mathrm{NN}}(W^{(n)}) \leq e^{-n( \min(E_{\mathrm{sp}}(\delta),E_{P}^{\mathrm{un}}(1+\varepsilon_n))+o(1))}  \label{BNN}
 \end{equation}
     \end{theorem}
     \begin{proof}
      Let $\beta :=\min(E_{\mathrm{sp}}(\delta),E_{P}^{\mathrm{un}}(1+\varepsilon)),$ and, for each $n\in\mathbb{Z}_{>1},$ denote $C_n:=\left[ \sqrt{n\sigma_{w,n}^2},\sqrt{n\sigma_{w,n}^2(1+\varepsilon)}\right].$ Note that, for each $n,$ equation \ref{ANN} and lemma \ref{IP} yield
     \begin{align*}
     &B_{\mathfrak{p}_n}^{\mathrm{NN}}(W^{(n)}) \leq \mathrm{Pr}( \|W^{(n)}\|>a_np_n/2)+ \left(\frac{\sqrt{\pi e/2}}{p_n^{1-k_n/n}} \right)^n \\
    &+n\left( \left( \frac{1}{\sqrt{1+\varepsilon}}+ \frac{\sqrt{\pi e/2}}{p_n^{1-k_n/n}} \right)^{n} +e^{-n \cdot \inf_{u\in C_n}v_{\mathfrak{p}_n}^{\mathrm{NN}}(u,\varepsilon)}\right).
    \end{align*}

Now, for each $n,$ we have that 
$$
a_np_n = V_{\mathfrak{p}_n}^{1/n}p_n^{k_n/n} = p_n^{k_n/n}\sqrt{2\pi e \sigma_{w,n}^2(1+\varepsilon)}> \sqrt{\delta n \sigma_{w,n}^2},
$$
so the Chernoff bound yields
     $$
     \mathrm{Pr}(\|W^{(n)}\|>a_np_n/2) < e^{-n E_{\mathrm{sp}}(\delta)}\leq e^{-n\beta}.
     $$
     
     On the other hand, the limit $\lim_{n\rightarrow \infty} p_n^{1-k_n/n}=\infty$ implies that $\left( \frac{\sqrt{\pi e/2}}{p_n^{1-k_n/n}} \right)^n<e^{-\beta n}$ for $n$ large enough, and, as $n\longrightarrow \infty,$
    \begin{align*}
    \log \left( \frac{1}{\sqrt{1+\varepsilon}} + \frac{\sqrt{\pi e/2}}{p_n^{1-k_n/n}} \right)^{-1} - \frac{\log n}{n}  &= \frac{1}{2}\log (1+\varepsilon) + o(1)
    \end{align*}
   Further, since, for any $b>0,$ $\frac{1}{2} \log(1+b) > E_{P}^{\mathrm{un}}(1+b),$ we have that $\frac{1}{2}\log(1+\varepsilon)>\beta.$ 
    
    Thus, inequality \ref{BNN} follows from lemma \ref{VNN}.
  \end{proof}
Note that $E_{\mathrm{sp}}$ maps $[1,\infty)$ bijectively into $[0,\infty).$ Hence, we may define a function $E_T: [1,\infty) \longrightarrow [1,\infty)$ such that $E_T(b)=E_{\mathrm{sp}}^{-1}(E_P^{\mathrm{un}}(b)).$ Note that $E_T(b)=b$ for $b\in [1,2],$ and $E_T(b)\leq b$ in general. Then, if $\delta$ in theorem \ref{PE} is chosen as $\delta=E_T(1+\varepsilon),$ we get that $B_{\mathfrak{p}_n}^{\mathrm{NN}}(W^{(n)}) \leq e^{-n(E_P^{\mathrm{un}}(1+\varepsilon)+o(1))}.$ 

\subsection{Flatness}

The following theorem gives sufficient conditions under which the flatness factor vanishes.

\begin{theorem} \label{FF} Let $\tau_1>\cdots>\tau_\ell>0,$ and $\lbrace \mathfrak{p}_n=(n,k_n,p_n,a_n)\rbrace_{n\in \mathbb{Z}_{>1}}$ be a sequence of quadruples of parameters. For each $j\in \lbrace 1,\cdots,\ell \rbrace,$ and $n\in \mathbb{Z}_{>1},$ let $f_j(n)$ and $g_j(n)$ be given by
   \begin{equation} \label{FG}
   a_n= \sqrt{\frac{\pi}{\tau_j \log(n/f_j(n))}}  \;\;\; \text{and} \;\;\;    p_n = \left( \frac{\log(n/g_j(n))}{\pi V_{\mathfrak{p}_n}^{2/n}\tau_j} \right)^{n/(2k_n)}
  \end{equation}
 If $\limsup_{n\rightarrow \infty} \tau_1 V_{\mathfrak{p}_n}^{2/n} <1,$  $\lim_{n\rightarrow \infty} p_n^{n-k_n}=\infty$ and $\max_j (f_j(n),g_j(n))=o(1)$ as $n\longrightarrow \infty,$  then $
 \sup_j \lim_{n\rightarrow \infty}  B^{\mathrm{Fl}}_{\mathfrak{p}_n}(W^{(n)},\tau_j)=0.$

\end{theorem}

\begin{proof}

   For each $j,$ the conditions given on $f_j$ and $g_j$ imply that  $n=o(e^{\pi/(a_n^2\tau_j)})$ and $n=o(e^{\pi a_n^2p_n^2 \tau_j}),$ respectively, so, by lemma \ref{RT}, $\lim_{n\rightarrow \infty} \Theta_{\mathbb{Z}^n}\left(\frac{1}{a^2\tau_j}\right)= 1=\lim_{n\rightarrow \infty}\Theta_{\mathbb{Z}^n}\left(a^2p^2\tau_j\right).$ Further, $\lim_{n\rightarrow \infty}p_n^{n-k_n}=\infty$ yields, by lemma \ref{PH} that $\lim_{n\rightarrow \infty} \xi_{\mathfrak{p}_n}p_n^{k_n} =1.$ Thus, $\limsup_{n\rightarrow \infty} \tau_1  V_{\mathfrak{p}_n}^{2/n}<1$ implies that, by definition of $B_{\mathfrak{p}_n}^{\mathrm{FL}}$ and expression \ref{AFL1}, $\lim_{n\rightarrow \infty} B^{\mathrm{Fl}}_{\mathfrak{p}_n}(W^{(n)},\tau_j)=0$ for each $j.$

\end{proof}

\subsection{Compatibility}

The usefulness of the results of theorems \ref{PE} and \ref{FF} hinge on the compatibility of the their premises. In this section, we show that the premises are compatible, i.e., that there exists a wide range of quadruples of parameters satisfying the premises in these theorems simultaneously.

We assume, for the remaining of the paper, that $\sigma_{w,n}$ is constant in $n,$ and set $\sigma_w=\sigma_{w,n}.$ 

\begin{theorem} \label{CM}  Let $\tau_1>\cdots>\tau_\ell>0$ be such that $2\pi e \sigma_{w}^2 \tau_1 <1,$ fix $b\in (2\pi e \sigma_{w}^2 , 1/\tau_1),$ and let $\delta'\geq 2/(\pi e).$ Then, for any sequence $\lbrace \mathfrak{p}_n = (n,k_n,p_n,a_n) \rbrace_{n\in \mathbb{Z}_{>1}}$ of quadruples of parameters with 
\begin{equation} \label{MAC}
 p_n > \max\left(  \left( \delta' n \right)^{n/(2k_n)},  \left( \frac{1}{\pi} \log n \right)^{n/(2(n-k_n))} \right)
 \end{equation}
 and $\lbrace V_{\mathfrak{p}_n}^{2/n} \rbrace_{n\in \mathbb{Z}_{>1}}\subset (2\pi e \sigma_{w}^2,b],$ we have that, with $f_j$ and $g_j$ as in (\ref{FG}),  $f_j(n)=o(1)$ and $g_j(n)=o(1)$ for every $j\in \lbrace 1,\cdots,\ell \rbrace$ as $n\longrightarrow \infty.$ 

 \end{theorem}
 
 \begin{proof} Note that the double sequence (in $n$ and in $j$) $\tau_j V_{\mathfrak{p}_n}^{2/n}$ is bounded away from $0;$ indeed, $\inf_{n,j} \tau_j V_{\mathfrak{p}_n}^{2/n} \geq 2\pi e \sigma_w^2 \tau_{\ell}>0.$ 
 
 Since $ p_n >  \left( \delta' n \right)^{n/(2k_n)},$  $g_j(n) < ne^{-\tau_j V_{\mathfrak{p}_n}^{2/n} \delta' n}$ so $g_j(n)=o(1).$ 
 Moreover, since $ p_n >  \left( \frac{1}{\pi} \log n \right)^{n/(2(n-k_n))}$ and $\tau_j V_{\mathfrak{p}_n}^{2/n}< \tau_1 b,$ writing $a_n=V_{\mathfrak{p}_n}^{1/n}/p_n^{1-k_n/n}$ we see that
 $$
 f_j(n) = ne^{-\pi p_n^{2(1-k_n/n)}/(\tau_j V_n^{2/n})} <n^{1-1/(\tau_1 b)}
 $$
 so also $f_j(n)=o(1).$ 
 \end{proof}
  \begin{remark} Note that, if $
 p_n > \left( \frac{1}{\pi} \log n \right)^{n/(2(n-k_n))},$ then $\lim_{n\rightarrow \infty} p_n^{1-k_n/n}=\infty.$ 
 \end{remark}

 \subsection{Application to Lattice Gaussian Coding}
 
 In \cite{Ling_Belf}, Ling and Belfiore introduce lattice Gaussian coding, and elegantly use the flatness factor to prove that this coding scheme can achieve the capacity of the  AWGN channel. 
 
 Let $\sigma_s>0,$ $c\in \mathbb{R}^n$ and $\Lambda$ be a lattice in $\mathbb{R}^n.$ Define $f_{\sigma_s,c}:\mathbb{R}^n\longrightarrow (0,\infty)$ by $f_{\sigma_s,c}(y)=e^{-\|y-c\|^2/(2\sigma_s^2)}/(2\pi \sigma_s^2)^{n/2},$ and set $f_{\sigma_s,c}(\Lambda)= \sum_{\lambda \in \Lambda} f_{\sigma_s,c}(\lambda)$ for short. Then, a lattice Gaussian random variable (over $\Lambda$ with a shift vector $c$ and parameter $\sigma_s$) is defined via its probability mass function $D_{\Lambda,\sigma_s,c}: \Lambda \longrightarrow (0,1),$ given by $D_{\Lambda,\sigma_s,c}(\lambda)=f_{\sigma_s,c}(\lambda)/f_{\sigma_s,c}(\Lambda).$
 
 If a signal $X$ is drawn according to $D_{\Lambda,\sigma_s,c},$ with $\epsilon_{\Lambda}\left( \sigma_s \right)<1,$ is used in an AWGN channel $Y=X+Z,$ where the noise $Z$ has variance $\sigma_z^2,$ Ling and Belfiore show that the probability of error under MAP decoding $P_e^{\mathrm{LG}}(\Lambda,\sigma_s,c;\sigma_z)$ can be upper bounded as
\begin{align*}
P_e^{\mathrm{LG}}(\Lambda,\sigma_s,c;\sigma_z)\leq \frac{1+ \epsilon_{\Lambda}\left( \widetilde{\sigma} \right)}{1- \epsilon_{\Lambda}\left( \sigma_s \right)} \cdot \mathrm{Pr}((\widetilde{\sigma}/\sigma_s)Z\not\in \mathcal{V}(\Lambda)),
\end{align*}
where $\widetilde{\sigma} = \sigma_s^2/\sqrt{\sigma_s^2+\sigma_z^2}.$ In the remaining of the paper, we set $\sigma_w=(\widetilde{\sigma}/\sigma_s)\sigma_z.$

Further, with $P=\frac{1}{n} \mathbb{E}[\|X-c\|^2],$ the entropy $\mathbb{H}(X)$ satisfies
\begin{align*}
\mathbb{H}(X) &= \log \left( (2\pi \sigma_s^2)^{n/2}f_{\sigma_s,c}(\Lambda)\right) + \frac{n}{2}\cdot \frac{P}{\sigma_s^2} \\
&\geq \log \left( \frac{\left( 1-\epsilon_{\Lambda}\left(\sigma_s \right)\right)(2\pi \sigma_s^2)^{n/2}}{\mu_L(\mathcal{V}(\Lambda))} \right)+ \frac{n}{2}\cdot \frac{P}{\sigma_s^2} 
\end{align*}

So, with $\mu(\mathcal{V}(\Lambda))^{2/n}=2\pi e \sigma_w^2(1+\varepsilon)$ and $\varepsilon>0,$ the maximum achievable rate $R_{\max}^{\mathrm{LG}}(\Lambda,\sigma_s,c;\sigma_z)$ satisfies
\begin{align*}
R_{\max}^{\mathrm{LG}}&(\Lambda,\sigma_s,c;\sigma_z) \geq \frac{1}{n} \mathbb{H}(X)  \\
&\geq \frac{1}{2}\log \left( \frac{\left( 1-\epsilon_{\Lambda}\left( \sigma_s \right)\right)^{2/n}}{(1+\varepsilon)e^{1-P/\sigma_s^2}} \right) + \frac{1}{2}\log \left( 1+ \frac{\sigma_s^2}{\sigma_z^2} \right).
\end{align*}

Fix a $t\in (0,\pi).$ Suppose that, for each $n\in \mathbb{Z}_{>1},$ $\Lambda^{(n)}$ is a lattice in $\mathbb{R}^n$ such that  $\epsilon_{\Lambda^{(n)}}\left( \sigma_s/\sqrt{\frac{\pi}{\pi-t}}\right)<1,$ and $c_n\in \mathbb{R}^n$ is any shift vector. For each $n,$ let $X^{(n)}$ be a random variable distributed according to $D_{\Lambda^{(n)},\sigma_s,c_n},$ and set $P_n = \frac{1}{n}\mathbb{E}[\|X^{(n)}-c_n\|^2].$  Ling and Belfiore show that $\lim_{n\rightarrow \infty} \frac{P_n}{\sigma_s}=1.$ In such a case, with $\mathrm{SNR}_n=P_n/\sigma_z^2,$ one has that for any $\varepsilon' >\frac{1}{2}\log(1+\varepsilon),$
$$
R_{\max}^{\mathrm{LG}}(\Lambda^{(n)},\sigma_s,c;\sigma_z) \geq \frac{1}{2} \log (1+\mathrm{SNR}_n) - \varepsilon'
$$
if $n$ is large enough.

The following theorem quantifies the primes needed for Ling and Belfiore's construction.

 \begin{theorem}
  Let $\lbrace c_n \rbrace_{n\in \mathbb{Z}_{>1}}$ be any sequence of shift vectors $c_n\in \mathbb{R}^n,$ and $\lbrace p_n \rbrace_{n\in \mathbb{Z}_{>1}}$ be any sequence of primes such that $p_n > (\delta'n)^{\frac{1}{2}\left( 1+ \frac{\log \log n}{\log n}\right)}$ where $\delta'=2/(\pi e).$  Assume that $\sigma_s^2/\sigma_z^2>e.$ Then, for any $\eta \in (0,\frac{1}{2}\log \sigma_s^2/(e\sigma_z^2))$ and any $\gamma\in(2\pi e,2\pi e^{1+2\eta}],$ there is a sequence of quadruples of parameters $\lbrace \mathfrak{p}_n=(n,k_n,p_n,a_n) \rbrace_{n\in \mathbb{Z}_{>1}}$ and a function $h,$ which satisfies  $h(n)=o(1)$ as $n\longrightarrow \infty,$ such that $\gamma_{\Lambda_{\mathfrak{p}_n}}(\sigma_w) =\gamma$ for every $n$ and as $n\longrightarrow \infty$
  \begin{align*}
  \mathrm{Pr}&\left\lbrace P_e^{\mathrm{LG}}(\Lambda_{\mathfrak{p}_n},\sigma_s,c_n;\sigma_z) \leq e^{-n(E_P^{\mathrm{un}}(\gamma/(2\pi e))+h(n))} \right., \\
 &\left. R_{\max}^{\mathrm{LG}}(\Lambda_{\mathfrak{p}_n},\sigma_s,c_n;\sigma_z)>\frac{1}{2}\log \left( 1+ \mathrm{SNR}_n \right) - \eta \right\rbrace \longrightarrow 1
  \end{align*}
  \end{theorem}
  \begin{proof}
  First, note that for each $n,$ there is a $k_n\in [n\log n /\log(n\log n),n-1]$ making $p_n$ satisfy inequality \ref{MAC} with $\delta'=2/(\pi e).$ Set $\tau_1=1/(2\pi \widetilde{\sigma}^2),$ $\tau_2=1/(2(\pi-t)\sigma_s^2)$ and $\tau_3 = 1/(2\pi\sigma_s^2),$ where $t$ is small enough so that $\tau_1>\tau_2>\tau_3>0.$ Note that $2\pi e \sigma_{w}^2 \tau_1 <1$ is equivalent to $\sigma_s^2/\sigma_z^2>e.$ For each $n,$ choose $a_n$ so that $\gamma_{\Lambda_{\mathfrak{p}_n}}(\sigma_w) =\gamma.$ Then, $\limsup_{n\rightarrow \infty} \tau_1 V_{\mathfrak{p}_n}^{2/n} \leq e^{1+2\eta}\sigma_{z}^2/\sigma_s^2<1.$ Then, theorem \ref{CM} yields that theorems \ref{PE} and \ref{FF} apply, and, in view of propositions \ref{PE1} and \ref{FF1}, Markov's inequality yields the desired result.
  \end{proof}

  \section*{Appendix A}
  
   We collect here some of the technical issues regarding measure theory.
   In this paper, we endow any finite set $T$ with the $\sigma-$algebra $\mathcal{P}(T),$ and a random variable over $T$ always refers to a $T-$valued measurable function.
   
   \subsubsection{} \label{AC} Consider any sets $T_1$ and $T_2,$ where $T_1$ is finite, any random variable $G$ over $T_1$ and any function $f:T_1 \longrightarrow T_2.$ Replacing $T_2$ by the range of $g,$ which is necessarily a finite set, it is clear that $f$ is measurable. Hence, $f(G)$ is a well-defined random variable. In this paper, whenever we consider the composition of random variables over finite sets with another function, we are assuming that a similar construction to the one discussed here is made. 
   
   \subsubsection{} \label{AT} Tonelli's theorem assures that the various interchanges of integrals made in this paper are justified. 
   
   \begin{theorem}[Tonelli] Let $(X,\Sigma_1,\mu)$ and $(Y,\Sigma_2,\nu)$ be $\sigma-$finite measure spaces, and $f:X\times Y \longrightarrow [0,\infty]$ be measurable. Then,
\begin{align*}
\int_X \int_Y f(x,y) \, d\nu \, d\mu &= \int_Y \int_X f(x,y) \, d\mu \, d\nu  \\
&= \int_{X\times Y} f(x,y) \, d\mu\times\nu.
\end{align*}
\end{theorem}
\begin{remark}
When $\nu$ is the counting measure, the theorem yields that $\int_X \sum_{y\in Y} f(x,y) \, d\mu = \sum_{y\in Y} \int_X  f(x,y) \, d\mu.$ If $\mu$ is also the counting measure, then the theorem yields that $\sum_{x\in X}  \sum_{y\in Y}  f(x,y) = \sum_{y\in Y}\sum_{x\in X}    f(x,y).$ Also, an extension yields that, when $X = \prod_{i=1}^n X_i$ is countable and $f:X\longrightarrow [0,\infty]$ is any function, $\sum_{i=\pi(1)}^{\pi(n)} \sum_{x_i\in X_i} f(x_1,\cdots,x_n) = \sum_{x\in X} f(x)$ for any permutation $\pi$ in the symmetric group $S_n.$ 
\end{remark}

  \subsubsection{} \label{AI}
 Denote the $n-$sphere by $\mathbb{S}^{n-1}.$ For a fixed $M\in \mathbb{F}_p^{n\times k},$ discreteness of $a\Lambda(M)$ implies that $|\mathcal{B}(0,2r)\cap a\Lambda(M)|<\infty.$ Thus, in particular, $\max_{w\in \mathbb{S}^{n-1}} \lbrace N_{\mathcal{B}(w,r)}(a\Lambda(M)\setminus ap\mathbb{Z}^n) \rbrace$ exists and is finite. Denote this maximum by $\ell.$ Let $f_M:r\mathbb{S}^{n-1} \longrightarrow \lbrace 0,\cdots,\ell \rbrace$ be defined by $f_M(w)=N_{\mathcal{B}(w,r)}(a\Lambda(M)\setminus ap\mathbb{Z}^n).$ By discreteness of $a\Lambda(M),$ $f_M^{-1}(\lbrace j \rbrace),$ for any $0\leq j \leq \ell,$ is a countable union of closed subsets of $r\mathbb{S}^{n-1}.$ In particular, each $f_M^{-1}(\lbrace j \rbrace)$ is measurable. Thus, for any Borel-measurable set $B\subset \mathbb{R}_{>0},$ the set $f_M^{-1}(B)=f^{-1}(B\cap \lbrace 0 ,\cdots,\ell \rbrace)= \bigcup_{j\in B_\ell} f_M^{-1}(\lbrace j \rbrace),$ where $B_\ell = B \cap \lbrace 0 ,\cdots,\ell \rbrace,$ is measurable. Hence, $f_M$ is a well-defined random variable, and $\mathbb{E}_{W^{(n)}}[f_M(W^{(n)})]$ is well-defined. Further, as, for any random variable $G$ over $\mathbb{F}_p^{n\times k},$ $\mathbb{E}_G[f_G(W^{(n)})]=\sum_{M\in \mathbb{F}_p^{n\times k}} \mathrm{Pr}(G=M) f_M(W^{(n)})$ is a finite sum, $\mathbb{E}_{W^{(n)}}[\mathbb{E}_G[f_G(W^{(n)})]]$ is also well-defined, and, by non-negativity of each $f_M(W^{(n)}),$ we may interchange the order of expectations.
 
  \section*{Appendix B} \label{B}

\noindent\textit{Proof (of Proposition \ref{av}).}
By lemma \ref{ineq}, we have that
\begin{align*}
  &\mathbb{E}_{G}\left[ \sum_{\lambda \in\Lambda(G)} g(\lambda) \right] = \sum_{M\in \mathbb{F}_p^{n\times k}} \mathrm{Pr}(G=M)\sum_{\lambda \in \Lambda(M)}  g(\lambda) \\
  &\leq  p^{k}  \sum_{y\in \mathbb{F}_{p}^{n}} \mathrm{Pr}(Gu_{\mathfrak{p}}=y) \sum_{z\in \mathbb{Z}^n} g(y+pz)\\
  &= p^k \mathbb{E}_{Gu_\mathfrak{p}}\left[ \sum_{z\in \mathbb{Z}^n} g(Gu_\mathfrak{p}+pz) \right],
\end{align*}
as desired \hfill $\Box$ \\

\noindent \textit{Proof (of Proposition \ref{Av1}).} Using  $g: \mathbb{R}^n \longrightarrow [0,\infty]$ defined by $g(\lambda)=N_S(\lbrace a\lambda\rbrace \setminus ap\mathbb{Z}^n)$ in proposition \ref{av}, one obtains 

\begin{align*}
   &\mathbb{E}_{G}\left[ N_{S}(a\Lambda(G)\setminus ap\mathbb{Z}^n) \right] \\
   &\leq p^k \mathbb{E}_{Gu_\mathfrak{p}}\left[ \sum_{z\in \mathbb{Z}^n} N_S(\lbrace a(Gu_\mathfrak{p}+pz) \rbrace\setminus ap\mathbb{Z}^n) \right] \\
   &= p^{k}  \sum_{y\in \mathbb{F}_{p}^{n}\setminus \lbrace 0 \rbrace} \mathrm{Pr}(Gu_{\mathfrak{p}}=y)\sum_{z\in \mathbb{Z}^n}  N_{S}(\lbrace a (y+pz) \rbrace) \\
   &\leq p^k \cdot \xi^M(G)  \sum_{y\in \mathbb{F}_{p}^n\setminus \lbrace 0 \rbrace}  N_{S}( a(y+p\mathbb{Z}^n ) ) \\
   &\leq  p^k \cdot \xi^M(G) \cdot  N_{S}(a\mathbb{Z}^n ).
   \end{align*}
Inequality \ref{EUNS} follows from \ref{EZNS} by substituting $G=U_\mathfrak{p}',$ using lemma \ref{U} and noting that $\xi^M(U_\mathfrak{p}')=\frac{1-\xi_\mathfrak{p}}{p^n-1}.$ \hfill $\Box$ \\

\noindent \textit{Proof (of Lemma \ref{NS}).} For the first statement, note that if $apx\in \mathcal{B}_n(q,r)$ for some $q\in \mathbb{R}^n$ and nonzero $x\in \mathbb{Z}^n,$ then $ap-\|q\|\leq \|apx\|-\|q\|\leq \|apx-q\|\leq r,$ so $ap-r\leq \|q\|.$ The second one follows from Markov's inequality. \hfill $\Box$

  \section*{Appendix C}\label{C}
  
  Before proving lemma \ref{IP}, observe that, since  $1/p^k\leq \xi_\mathfrak{p}$ and $k\leq n-1,$ we have 
  \begin{align}
  &(1-\xi_\mathfrak{p})p^n/((1-p^{k-n})(p^n-1)) \nonumber \\
  &\leq p^2(p^{n-1}-1)/((p-1)(p^n-1)) <p/(p-1)\leq 2 \label{2}
   \end{align}
Also, recall the following well-known result.

\begin{theorem} \label{VUB}
 For each $n\in \mathbb{N},$ we have that $\mu_L(\mathcal{B}_n(0,1)) \leq \frac{1}{\sqrt{n\pi}} \left( \frac{2\pi e}{n} \right)^{n/2}.$ Furthermore, as $n\longrightarrow \infty,$  $\mu_L(\mathcal{B}_n(0,1)) \sim \frac{1}{\sqrt{n\pi}} \left( \frac{2\pi e}{n} \right)^{n/2}.$   
 \end{theorem}
  
\noindent \textit{Proof (of Lemma \ref{IP}).} First, inequalities \ref{EUNS} and \ref{2}, lemma \ref{OR} and theorem \ref{VUB} yield that, for any $\rho>0,$
 \begin{equation}
 h(\mathfrak{p},U_{\mathfrak{p}}, \rho) \leq \frac{2}{\sqrt{\pi n}} \left( \frac{r\sqrt{2\pi e}}{\sqrt{n}V_{\mathfrak{p}}^{1/n}} + \frac{\sqrt{\pi e/2}}{p^{1-k/n}} \right)^n \label{ineq2}
 \end{equation}
 
Then,
    \begin{align*}
    &I_{\mathfrak{p}}\left(W^{(n)},\Lambda_{\mathfrak{p}}\right)\leq \int_0^\infty f_{\|W^{(n)}\|}(r) \\
    &\cdot \min\left( \frac{2}{\sqrt{\pi n}} \left( \frac{r\sqrt{2\pi e}}{\sqrt{n}V_{\mathfrak{p}}^{1/n}} + \frac{\sqrt{\pi e/2}}{p^{1-k/n}} \right)^n,1 \right) \, dr  \\
     &< \int_0^{\sqrt{n\sigma_{w,n}^2(1+\varepsilon)}} f_{\|W^{(n)}\|}(r) \\
     &\cdot \left( \frac{r}{\sqrt{n\sigma_{w,n}^2(1+\varepsilon)}} + \frac{\sqrt{\pi e/2}}{p^{1-k/n}} \right)^n \, dr \\
     &+ \mathrm{Pr}\left(\|W^{(n)}\|>\sqrt{n\sigma_{w,n}^2(1+\varepsilon)}\right).
     \end{align*}

Now, integration by parts yields that, since $\frac{\partial}{\partial r} \left( - \mathrm{Pr}(\|W^{(n)}\|>r) \right) = f_{\|W^{(n)}\|}(r,n),$ 
     \begin{align*}
     &\int_0^{\sqrt{n\sigma_{w,n}^2(1+\varepsilon)}} f_{\|W^{(n)}\|}(r) \left( \frac{r}{\sqrt{n\sigma_{w,n}^2(1+\varepsilon)}} + \frac{\sqrt{\pi e/2}}{p^{1-k/n}} \right)^n \, dr \\
     &= - \mathrm{Pr}\left(\|W^{(n)}\|>\sqrt{n\sigma_{w,n}^2(1+\varepsilon)} \right)\left( 1+ \frac{\sqrt{\pi e/2}}{p^{1-k/n}} \right)^n\\
     &+\left(\frac{\sqrt{\pi e/2}}{p^{1-k/n}} \right)^n +  \sqrt{\frac{n}{\sigma_{w,n}^2 (1+\varepsilon)}} J_{0,\sqrt{n\sigma_{w,n}^2(1+\varepsilon)}},
     \end{align*}
     where
    \begin{align*}
  J_{\alpha,\alpha'}:=    \int_{\alpha}^{\alpha'} &\mathrm{Pr}(\|W^{(n)}\|>r) \\
     &\cdot \left( \frac{r}{\sqrt{n\sigma_{w,n}^2(1+\varepsilon)}} + \frac{\sqrt{\pi e/2}}{p^{1-k/n}} \right)^{n-1} \, dr
    \end{align*}
     Now, note that 
      \begin{align*}
    J_{0,\sqrt{n\sigma_{w,n}^2}}\leq \sqrt{n\sigma_{w,n}^2} \left( \frac{1}{\sqrt{1+\varepsilon}} + \frac{\sqrt{\pi e/2}}{p^{1-k/n}} \right)^{n-1},
     \end{align*}
     and, by lemma \ref{CB},
      \begin{align*}
    J_{\sqrt{n\sigma_{w,n}^2},\sqrt{n\sigma_{w,n}^2(1+\varepsilon)}} &\leq \int_{\sqrt{n\sigma_{w,n}^2}}^{\sqrt{n\sigma_{w,n}^2(1+\varepsilon)}} e^{-nE\left( \frac{r^2}{n\sigma_{w,n}^2} \right)} \\
    &\cdot \left( \frac{r}{\sqrt{n\sigma_{w,n}^2(1+\varepsilon)}} + \frac{\sqrt{\pi e/2}}{p^{1-k/n}} \right)^{n-1} \, dr \\
       &= \int_{\sqrt{n\sigma_{w,n}^2}}^{\sqrt{n\sigma_{w,n}^2(1+\varepsilon)}} e^{-n\cdot v_{\mathfrak{p}}^{\mathrm{NN}}(r,\varepsilon)} \, dr \\
    &\leq  \sqrt{n\sigma_{w,n}^2(1+\varepsilon)} e^{-n \cdot \inf_{u\in C}v_{\mathfrak{p}}^{\mathrm{NN}}(u,\varepsilon)}.
    \end{align*}
    Hence, we have that
     \begin{align*}
    &I_{\mathfrak{p}}\left(W^{(n)},\Lambda_{\mathfrak{p}}\right)< \left(\frac{\sqrt{\pi e/2}}{p^{1-k/n}} \right)^n \\
    &+n\left( \left( \frac{1}{\sqrt{1+\varepsilon}}+ \frac{\sqrt{\pi e/2}}{p^{1-k/n}} \right)^{n} +e^{-n \cdot \inf_{u\in C}v_{\mathfrak{p}}^{\mathrm{NN}}(u,\varepsilon)}\right).
    \end{align*}
\hfill $\Box$

   Before proving lemma \ref{VNN}, note that the function  $v_{\mathfrak{p}}^{\mathrm{NN}}(\cdot,\varepsilon)$ is strictly convex over $[0,\infty).$ Indeed, for any $u\in \mathbb{R}_{\geq 0},$
    $$
    \frac{\partial^2}{\partial u^2} v_{\mathfrak{p}}^{\mathrm{NN}}(u,b) = \frac{1}{n\sigma_{w,n}^2} + \frac{1}{u^2} + \frac{(n-1)/n}{\left( u + \frac{\sqrt{n\sigma_{w,n}^2\pi e/2}}{p_n^{1-k_n/n}}\right)^2}>0.
    $$
    In particular, $v_{\mathfrak{p}}^{\mathrm{NN}}(\cdot,b)$ has a unique minimum over any bounded closed subinterval of $[0,\infty).$ \\

\noindent \textit{Proof (of Lemma \ref{VNN}).}
    For each $n,$  define $f_{n,b}:\mathbb{R}_{\geq 0}\longrightarrow \mathbb{R}$ by 
    \begin{align*}
    f_{n,b}(y)&=\left.\frac{\partial}{\partial u}v_{\mathfrak{p}_n}^{\mathrm{NN}}\left( u,b\right)\right|_{u=y} \\
    &= \frac{y}{n\sigma_{w,n}^2} -\frac{1}{y} - \frac{(n-1)/n}{ y + \frac{\sqrt{\pi en\sigma_{w,n}^2(1+b)/2}}{p_n^{1-k_n/n}} }
    \end{align*}
    and denote $C_n:=\left[ \sqrt{n\sigma_{w,n}^2},\sqrt{n\sigma_{w,n}^2(1+b)}\right]$ and $u_n:= \mathrm{argmin}_{u\in C_n} v_{\mathfrak{p}_n}^{\mathrm{NN}}\left( u,b\right).$ Note that, for each $n,$
    \begin{align}
       f_{n,b}\left(\sqrt{n\sigma_{w,n}^2(1+b)}\right)&=\frac{1}{\sqrt{n\sigma_{w,n}^2(1+b)}}\left( b - \frac{(n-1)/n}{1 + \frac{\sqrt{\pi e/2}}{p_n^{1-k_n/n}} }\right) \label{V2} \\
    f_{n,b}\left( \sqrt{2n\sigma_{w,n}^2}\right)&= \frac{1}{\sqrt{2n\sigma_{w,n}^2}} \left( 1 - \frac{(n-1)/n}{1+ \frac{\sqrt{\pi e (1+b)}}{2p_n^{1-k_n/n}}} \right)>0 \label{V3}
    \end{align}
    If $b<1,$ then equation $\ref{V2}$ implies that $u_n= \sqrt{n\sigma_{w,n}^2(1+b)}$ for all large $n.$ As 
    $$
    v_{\mathfrak{p}_n}^{\mathrm{NN}}\left( \sqrt{n\sigma_{w,n}^2(1+b)},b \right) =E_{\mathrm{sp}}(1+b) - \frac{n-1}{n} \log \left( 1+ \frac{\sqrt{\pi e/2}}{p_n^{1-k_n/n}} \right)
    $$
    we see that $\inf_{u\in C_n} v_{\mathfrak{p}_n}^{\mathrm{NN}}(u,b) = E_{P}^{\mathrm{un}}(1+b)+o(1)$ in this case.
    
    Now, assume that $b\geq 1.$ Then, for each $n,$ $\sqrt{2n\sigma_{w,n}^2}\in C_n,$ so inequality \ref{V3} implies that $\delta_n:=\frac{u_n}{\sqrt{2n\sigma_{w,n}^2}}<1.$ On the other hand, the sequence $\lbrace \alpha_n:=1- 1/\min(n,1+p_n^{1-k_n/n}/\sqrt{\pi e(1+b)})\rbrace_{n\in \mathbb{Z}_{>1}}\subset (0,1)$ satisfies $\lim_{n\rightarrow \infty}\alpha_n=1$ and $\alpha_n<\delta_n$ for all large $n$; indeed, for all large $n,$ we have that $2\alpha_n>1,$ so
    \begin{align*}
    &\alpha_n \sqrt{2n\sigma_{w,n}^2} f_{n,b}\left(\alpha_n\sqrt{2n\sigma_{w,n}^2}\right)\\
    &= 2\alpha_n^2-1-\frac{(n-1)/n}{1+\frac{\sqrt{\pi e (1+b)}}{2\alpha_n p_n^{1-k_n/n}}} < 2\alpha_n^2-1-\frac{(n-1)/n}{1+\frac{\sqrt{\pi e (1+b)}}{ p_n^{1-k_n/n}}} \\
    &=2\alpha_n^2 - 1-\left(1 - \frac{1}{n} \right)\left( 1- \frac{1}{\frac{p_n^{1-k_n/n}}{\sqrt{\pi e (1+b)}}+1}  \right)\leq \alpha_n^2-1<0.
    \end{align*}
    Thus, we have that $\lim_{n\rightarrow \infty}\delta_n=1,$ and 
    \begin{align*}
    &\lim_{n\rightarrow \infty}v_{\mathfrak{p}_n}^{\mathrm{NN}}(\delta_n\sqrt{2n\sigma_{w,n}^2}) \\
    &=  \lim_{n\rightarrow \infty} E_{\mathrm{sp}}(2\delta_n^2) - \frac{n-1}{n} \log \left( \frac{\delta_n \sqrt{2}}{\sqrt{1+b}} + \frac{\sqrt{\pi e/2}}{p_n^{1-k_n/n}} \right) \\
    &= \frac{1}{2}\log \frac{e(1+b)}{4},
    \end{align*}
   so $\inf_{u\in C_n} v_{\mathfrak{p}_n}^{\mathrm{NN}}(u,b) = E_{P}^{\mathrm{un}}(1+b)+o(1)$ in this case too. \hfill $\Box$

\section*{Appendix D} \label{D}

Note that, for any quadruple of parameters $\mathfrak{p}=(n,k,p,a),$ 
\begin{align}
\xi_{\mathfrak{p}}&= \sum_{j=0}^{k} \mathrm{Pr}(U_{\mathfrak{p}}'u_{\mathfrak{p}}=0 \, | \,\mathrm{rank}(U_{\mathfrak{p}}')=j) \mathrm{Pr}(\mathrm{rank}(U_{\mathfrak{p}}')=j) \nonumber \\
&=  \sum_{j=0}^{k}\frac{1}{p^j} \mathrm{Pr}(\mathrm{rank}(U_{\mathfrak{p}}')=j) \label{TE}
\end{align}
Before proving lemma \ref{PH}, we analyze the term $\mathrm{Pr}(\mathrm{rank}(U_{\mathfrak{p}}')=j),$ for which the following notation is convenient.

 \begin{definition}[$q-$Pochhammer Symbol]
For any $(a,q,n) \in \mathbb{R}\times \mathbb{R}\times (\mathbb{Z}\cup \lbrace \infty \rbrace),$ the $q-$Pochhammer symbol $(a;q)_n$ is defined by
$$
(a;q)_n = \left\lbrace \begin{array}{ll}
\prod_{\ell=0}^{n-1} (1-aq^\ell), &\text{if } n\geq 0 \\
\prod_{\ell=n}^{-1} (1-aq^\ell), &\text{otherwise}
\end{array} \right.
$$
whenever the product converges, and where the empty product is taken to be $1.$ When $a=q$ and $n=\infty,$ one obtains the Euler function $\phi(q)=(q;q)_\infty.$ 
\end{definition}
\begin{remark}
The Euler function will be of interest to us when $1/q$ is a prime number. One can show that, if $|q|<1,$ the Euler function is well-defined and nonzero. This is clear for $q=0,$ so assume $|q|<1$ and $q\neq 0.$ The product $\phi(q)$ is well-defined and nonzero if and only if the sum $S:=\sum_{\ell=1}^\infty \ln(1-q^\ell)$ converges. But, for each positive integer $\ell,$ we have the Taylor expansion $ \ln(1-q^\ell)=\sum_{m=1}^\infty \frac{q^{\ell m}}{m},$ so Tonelli's theorem yields that
\begin{align*}
\sum_{\ell=1}^\infty &\left|\sum_{m=1}^\infty \frac{q^{\ell m}}{m}\right|\leq \sum_{\ell=1}^\infty \sum_{m=1}^\infty \left|\frac{q^{\ell m}}{m}\right|= \sum_{m=1}^\infty \frac{1}{m} \sum_{\ell=1}^\infty |q|^{\ell m} \\
&= \sum_{m=1}^\infty \frac{1}{m(|q|^{-m}-1)} \leq \sum_{m=1}^\infty \frac{|q|^m}{1-|q|}  = \frac{|q|}{(1-|q|)^2}<\infty.
\end{align*}
Then, $S$ is absolutely convergent, so $\phi(q)$ is well-defined and nonzero. Further, $\min_{p\in \mathscr{P},n\geq 0} (1/p ; 1/p)_n = \phi(1/2)>e^{-2}.$ In fact, one may show that $ \phi(1/2)=0.288788\hdots$ .
\end{remark}

The following is a well-known fact. 

 \begin{lemma} \label{RU1}
 Fix a quadruple of parameters $\mathfrak{p}=(n,k,p,a).$  Then, $\mathrm{Pr}(\mathrm{rank}(U_\mathfrak{p}')=0)=1/p^{nk},$ and, for any integer $1\leq j \leq k,$
  \begin{align*}
  &\mathrm{Pr}(\mathrm{rank}(U_{\mathfrak{p}}')=j)\\
  &= \frac{1}{p^{nk}} \cdot \frac{(p^n-1)\cdots (p^n-p^{j-1})\cdot (p^k-1)\cdots (p^{k-(j-1)}-1)}{(p-1)\cdots (p^j-1)} \\
  &= \frac{1}{p^{(n-j)(k-j)}} \cdot \frac{\left( \frac{1}{p} ; \frac{1}{p} \right)_n\left( \frac{1}{p} ; \frac{1}{p} \right)_k}{\left( \frac{1}{p} ; \frac{1}{p} \right)_j\left( \frac{1}{p} ; \frac{1}{p} \right)_{n-j}\left( \frac{1}{p} ; \frac{1}{p} \right)_{k-j}}.
  \end{align*}
  \end{lemma}
  
  Using lemma \ref{RU1}, one may prove the following bounds on $\mathrm{Pr}(\mathrm{rank}(U_{\mathfrak{p}}')=j).$
  
\begin{lemma}\label{RU}
 For any quadruple of parameters $\mathfrak{p}=(n,k,p,a),$ and any $1\leq j\leq k-1,$ we have that
  \begin{equation}\label{RU2}
  \mathrm{Pr}(\mathrm{rank}(U_{\mathfrak{p}}')=j) < \frac{1}{p^{(n-k+1)(k-j)}\phi(1/2)}
  \end{equation}
  Also, $  \mathrm{Pr}(\mathrm{rank}(U_{\mathfrak{p}}')=k)> 1- p^{k-n}.$ 
  \end{lemma}
 \begin{proof}
 Note that, for any $m< \ell,$ $\phi(1/2)\leq\left( \frac{1}{p} ; \frac{1}{p} \right)_{\ell} < \left( \frac{1}{p} ; \frac{1}{p} \right)_{m}.$ Thus, for any $1\leq j \leq k-1,$
 $$
 \frac{\left( \frac{1}{p} ; \frac{1}{p} \right)_n\left( \frac{1}{p} ; \frac{1}{p} \right)_k}{\left( \frac{1}{p} ; \frac{1}{p} \right)_j\left( \frac{1}{p} ; \frac{1}{p} \right)_{n-j}\left( \frac{1}{p} ; \frac{1}{p} \right)_{k-j}} < \frac{1}{\phi(1/2)}.
 $$ 
 Then, lemma \ref{RU1} yields \ref{RU2}.
 
 For each $1\leq j\leq k,$ let $\mathfrak{p}_j=(n,j,p,a),$ and note that $\mathfrak{p}_k=\mathfrak{p}$ and $\delta_j:=\mathrm{Pr}(\mathrm{rank}(U_{\mathfrak{p}_j}')=j)=\frac{\left( \frac{1}{p} ; \frac{1}{p} \right)_n}{\left( \frac{1}{p} ; \frac{1}{p} \right)_{n-j}}.$ We will show that $\delta_k> 1-p^{k-n}.$  First, note that $2\leq p$ implies that $\frac{2}{p} - \frac{1}{p^{n-j+1}}<1,$ so 
 $$
 1 - \frac{2}{p^{n-j}}+\frac{1}{p^{2(n-j)}} > 1 -\frac{1}{p^{n-j-1}},
 $$
 or, $(1-p^{j-n})^2>(1-p^{j+1-n})$ for any $j.$
Thus, for any $1\leq j \leq k-1,$ if we have that $\delta_j> 1-p^{j-n},$ we would also have 
$$
\delta_{j+1}=  (1-p^{j-n})\delta_j > (1-p^{j-n})^2>(1-p^{j+1-n}).
$$
As $\delta_1 = 1-p^{-n} > 1-p^{1-n},$ we see that $\delta_k > (1-p^{k-n}),$ as desired.
 \end{proof}
  
 \textit{Proof (of Lemma \ref{PH}).}
  For each $n,$ we have that $\frac{1}{p_n^{k_n}} \leq \xi_{\mathfrak{p}_n}$ and, by equation \ref{TE} and inequality \ref{RU2},
  \begin{align*}
\xi_{\mathfrak{p}_n}&< \frac{1}{p_n^{nk_n}} + \frac{1}{p_n^{k_n}}\left( 1+ \frac{1}{\phi(1/2)}\sum_{j=1}^{k_n-1} \frac{1}{\left(p_n^{(n-k_n)}\right)^{k_n-j}}\right) \\
   &< \frac{1}{p_n^{nk_n}} + \frac{1}{p_n^{k_n}}\left( 1+ \frac{1}{\phi(1/2)}\sum_{j=1}^{\infty} \frac{1}{\left(p_n^{(n-k_n)}\right)^{j}}\right) \\
  &= \frac{1}{p_n^{nk_n}} + \frac{1}{p_n^{k_n}}\left( 1+ \frac{1}{\phi(1/2)(p_n^{n-k_n}-1)}\right),
  \end{align*}
  so the desired result follows. \hfill $\Box$

 \section*{Appendix E}\label{E}

\textit{Proof (of Lemma \ref{U}).} Note that, for any $M\in \mathbb{F}_p^{n\times k},$ 
$$
\mathrm{Pr}(U_\mathfrak{p}'=M|U_\mathfrak{p}'\in M_\mathfrak{p})=\frac{1}{|M_\mathfrak{p}|}\cdot 1_{M_\mathfrak{p}}(M)=\mathrm{Pr}(U_\mathfrak{p}=M).
$$
 Hence, lemma \ref{RU} implies that
\begin{align*}
 \mathbb{E}_{U_\mathfrak{p}'}\left[ f(U_\mathfrak{p}')\right] &=   \mathbb{E}_{U_\mathfrak{p}'}\left[ f(U_\mathfrak{p}')|U_\mathfrak{p}'\in M_\mathfrak{p}\right]\mathrm{Pr}(U_\mathfrak{p}'\in M_\mathfrak{p}) \\
 &+ \mathbb{E}_{U_\mathfrak{p}'}\left[ f(U_\mathfrak{p}')|U_\mathfrak{p}'\not\in M_\mathfrak{p}\right]\mathrm{Pr}(U_\mathfrak{p}'\not\in M_\mathfrak{p}) \\
 &\geq \mathbb{E}_{U_\mathfrak{p}'}\left[ f(U_\mathfrak{p}')|U_\mathfrak{p}'\in M_\mathfrak{p}\right] (1-p^{k-n}) \\
 &= \mathbb{E}_{U_\mathfrak{p}}\left[ f(U_\mathfrak{p})\right](1-p^{k-n}).
\end{align*}
\hfill $\Box$

\section*{Appendix F}\label{F}

Note that, for any positive integer $n$ and $\tau>0,$
$$
\Theta_{\mathbb{Z}^{n}}(\tau)=\left( \theta(0,i\tau) \right)^n,
$$
where $\theta(0,i\tau):=\sum_{z\in \mathbb{Z}} e^{-\pi \tau z^2}$ is the Jacobi theta function. \\

\textit{Proof (of Lemma \ref{RT}).} First, note that, for every $n,$
  \begin{align*}
 1+\frac{2}{e^{\pi c_n}}&< \theta(0,ic_n) = 1+ 2 \sum_{z=1}^\infty e^{-\pi c_n z^2} <1+ 2\sum_{z=1}^\infty e^{-\pi c_n z} \\
 &= 1+ \frac{2}{e^{\pi c_n}-1}.
  \end{align*}
  Now, assume that $n=o(e^{\pi c_n})$ as $n\longrightarrow \infty.$ Then, $\lim_{n\rightarrow \infty} \frac{n}{e^{\pi c_n}-1}=0,$ and for all large $n$ 
  \begin{align*}
  1 &< (\theta(0,ic_n))^n < \left(\left(  1+ \frac{2}{e^{\pi c_n}-1} \right)^{(e^{\pi c_n}-1)/2}\right)^{2n/(e^{\pi c_n}-1)} \\
  &< e^{2n/(e^{\pi c_n}-1)}.
  \end{align*}
  Hence, $\lim_{n\rightarrow \infty} (\theta(0,ic_n))^n =1.$ 
  
  For the converse, assume that $\lim_{n\rightarrow \infty} (\theta(0,ic_n))^n =1.$ Then, 
  $$
  \lim_{n\rightarrow \infty} \left( \left( 1+ \frac{2}{e^{\pi c_n}} \right)^{e^{\pi c_n}/2}\right)^{2n/e^{\pi c_n}}=1.
  $$
  Thus, $\lim_{n\rightarrow \infty} c_n=\infty.$ Hence, for all large $n,$
  $$
  2^{2n/e^{\pi c_n}}<\left( \left( 1+ \frac{2}{e^{\pi c_n}} \right)^{e^{\pi c_n}/2}\right)^{2n/e^{\pi c_n}},
  $$
  implying that $n=o(e^{\pi c_n}).$ \hfill $\Box$

\end{document}